\documentclass{article}
\usepackage{amsmath, amssymb, amsfonts, amsthm}
\usepackage{graphicx}
\usepackage{hyperref}
\usepackage{tikz}
\usepackage{tikz-qtree}
\usepackage[margin=2.5cm]{geometry}

\newtheorem{corollary}{Corollary}
\newcommand{\classX}[1]{\ensuremath{\text{\textsf{\textbf{#1}}}}} 

\newcommand{\classNP}{\classX{NP}}
\newcommand{\NPC}{\classX{NP-complete}}

\newcommand{\NPH}{\classX{NP-hard}}

\newcommand{\MAD}{\ensuremath{\text{\textit{MAD}}}}
\newcommand{\Prob}[1]{\ensuremath{\text{\textsc{#1}}}}

\newtheorem{theorem}{Theorem}
\newtheorem{lemma}{Lemma}
\newtheorem{proposition}{Proposition}

\begin{document}

\title{On graphs with well-distributed edge density}
\author{Syed Mujtaba Hassan\footnote{\url{ms06948@st.habib.edu.pk}, Computer Science Department, Habib University, Karachi, Pakistan (corresponding author)} \and Shahid Hussain\footnote{\url{shahidhussain@iba.edu.pk}, Computer Science Department, Institute of Business Administration, Karachi, Pakistan}}
\date{March 21, 2024}
\maketitle

\begin{abstract}
    In this paper, we introduce a class of graphs which we call \emph{average hereditary graphs}. Many graphs that occur in the usual graph theory applications belong to this class of graphs. Many popular types of graphs fall under this class, such as regular graphs, trees and other popular classes of graphs. The paper aims to explore some interesting properties regarding colorings average hereditary graphs. We prove a new upper bound for the chromatic number of a graph in terms of its maximum average degree and show that this bound is an improvement on previous bounds. From this, we show a relationship between the average degree and the chromatic number of an \textit{average hereditary graph}. We then show that even with new bound, the graph \Prob{3-coloring} problem remains $\NPH$ when the input is restricted to \emph{average hereditary graphs}. We provide an equivalent condition for a graph to be \textit{average hereditary}, through which we show that we can decide if a given graph is \textit{average hereditary} in polynomial time. \smallskip

\end{abstract}

\renewcommand\thefootnote{}

\renewcommand\thefootnote{\fnsymbol{footnote}}
\setcounter{footnote}{1}
\section{Introduction}\label{sec: intro}
We introduce a new class of graphs that we call average hereditary graphs. There are many interesting properties regarding colorings graphs of this class, a subset of which we explore in this paper. We provide a new upper bound for the chromatic number of a graph in terms of its maximum average degree, that is $\chi(G) \leq \MAD(G) + 1$, we show that this bound improves upon previously known results. We use this bound to find a relationship between the average degree of a graph and its chromatic number, from this, we prove a general case for which we know the exact chromatic number of an average hereditary graph. However, despite these results the graph \Prob{3-coloring} problem remains $\NPH$ when the input is restricted to average hereditary graphs. We show this in section~\ref*{sec: hardness}. This result is of interest to us in studying the complexity dichotomy of graph \Prob{3-coloring} problem.

We will analyze the computational complexity of deciding if a graph is average hereditary, and show that we can compute if a graph is average hereditary in polynomial time. This is done by providing an equivalent condition for a graph to be average hereditary. What's interesting is that the class of average hereditary graphs is quite broad as it contains many popular classes of graphs such as trees, regular graphs, and others commonly encountered in applications. So the class only restricts some extreme cases. 

The main result of our paper is the new general case upper bound on the chromatic number of graphs which is an improvement on previous bounds found in literature. 

\subsection{Preliminaries}\label{sec: prelim}
Throughout this document, we will denote a graph as $G = (V, E)$ where $G$ is a graph with vertex set $V$ and edge set $E$. An edge is represented as a set $\{v,u\}$ for some $v$ and $u$ belonging to $V$ such that $v \neq u$. Also, for a graph $G$ we use $V(G)$ to denote the vertex set of $G$ and $E(G)$ to denote the edge set of $G$. All graphs considered here are undirected and simple meaning they contain no loops, multi-edges or directed edges. We will use $H \subseteq G$ to denote an induced subgraph of $G$. For some $U \subseteq V$, $G-U$ denotes the subgraph of $G$ obtained by removing the vertices in $U$. For $H \subseteq G$, \;$\overline{H}$ denote the graph $G-V(H)$. For each vertex $v\in V$ we used $d_G(v)$ to denote the degree of $v$ in $G$. We also use $\Delta(G)$ to denote the maximum degree of $G$ and $\delta(G)$ to denote the minimum degree of $G$. The average degree of $G$, denoted by $d(G)$ is the average of all the degrees of $G$, which can be computed by $d(G) = \frac{\sum_{v\in V}d_G(v)}{|V|}= \frac{2|E|}{|V|}$ if $V$ is nonempty. For a null graph, we define the average degree as 0. If the degrees of all the vertices of a graph $G$ are equal to $k$, we say $G$ is $k$-regular.
 
The edge cut $[V(H),\overline{V(H)}]$ is the smallest set of edges you need to remove from $G$ to break $G$ into two components $H$ and $\overline{H}$, for $H \subseteq G$. The edge connectivity is the size of the smallest cut edge, denoted by $\kappa'(G)$. 
The clique size is the size of the largest complete subgraph of $G$, we denote it by $\omega(G)$. 

Coloring a graph $G$ is assigning a color to each vertex of $G$ such that if two vertices are adjacent then they are assigned a different color than each other. The smallest natural number $k$ such that $G$ can be colored with $k$ colors is known as the chromatic number of $G$. We denote the chromatic number of $G$ by $\chi(G)$.
If $k = \chi(G)$ then we say $G$ is $k$-chromatic. $G$ is called $k$-critical if $\chi(G) = k$ and $\forall v \in V,\; \chi(G-v) < k$. $\mathbb{N}$ denotes the set of natural numbers which includes $0$.

\subsection{Average Hereditary Graphs}\label{sec: avg}
We now define a new class of graphs, which we call \emph{average hereditary graphs}, and demonstrate that this class is broad, in other words, we show that this class of graph is a superset of many other popular classes of graphs.

We define a new class of graphs called average hereditary graphs. A graph $G = (V,E)$ is \emph{average hereditary} if for every induced subgraph $H$ of $G$, $d(H) \leq d(G)$. 
Note that this definition is equivalent to requiring that $d(H) \leq d(G)$ for all (not necessarily induced) subgraphs $H \subseteq G$. To see this, observe that if $d(H) \leq d(G)$ holds for all subgraphs, then in particular it holds for all induced subgraphs. Conversely, if the inequality holds for every induced subgraph $H' \subseteq G$, then for any subgraph $H \subseteq G$, there exists an induced subgraph $H' \subseteq G$ such that $V(H) = V(H')$ and $d(H) \leq d(H')$, implying $d(H) \leq d(G)$.


One way to think about such graphs is that the edge density is more ``uniformly'' distributed over the graph. In practice, most graphs arising in applications or standard graph-theoretic contexts belong to this class. The definition only restricts some extreme cases. Most graphs we see are average hereditary. Many popular classes of graphs belong to this class of graphs. For example, all regular graphs and all trees are average hereditary. The following propositions prove these claims.

\begin{proposition}\label{tree}
    If $G$ is a tree then $G$ is average hereditary.     
\end{proposition}

\begin{proof}
    Let $G = (V,E)$ be a tree so we have that $|E| = |V| - 1$. Therefore $d(G) =  \frac{2|E|}{|V|} = \frac{2  (|V|-1)}{|V|} = 2 - \frac{2}{|V|}$. Now if we induce a subgraph $H$ from $G$ then we can have two cases.
    First, if $H$ is empty then $d(H) \leq d(G)$. If $H$ is nonempty then we have that $H$ is a forest. So we have that $|E(H)| \leq |V(H)| - 1$, therefore $d(H) \leq 2 - \frac{2}{|V(H)|}$ and as both $V$ and $V(H)$ are finite, we have that $d(H) \leq d(G)$. 
\end{proof}

\begin{proposition}\label{reg}
    If $G$ is a $k$-regular graph then $G$ is average hereditary.     
\end{proposition}

\begin{proof}
    Let $G = (V, E)$ be a $k$-regular graph. So $d(G) = k$.
    Now if we remove any set of vertices from $G$ to induce a subgraph $H$ then every vertex in $H$ will have a degree of at most $k$. So the average degree of $H$ is at most $k$. Therefore $G$ will be average hereditary. 
    
\end{proof}

These propositions show that many popular types of graphs such as cycle, paths, star, claws, complete graphs, and Peterson graph, are all average hereditary, as we have that regular graphs and trees are average hereditary in general.
Figure~\ref{fig:exampl_ahg} shows some examples of average hereditary graphs.
\begin{figure}[h]
    \begin{center}
    \begin{tikzpicture}[scale=.75]
        \node[style={fill=black,circle}] (1) at (0,0){};
        \node[style={fill=black,circle}] (2) at (0,1.5){};
        \node[style={fill=black,circle}] (3) at (1.5,0){};
        \node[style={fill=black,circle}] (4) at (-1.5,0){};
        \node[style={fill=black,circle}] (5) at (0,-1.5){};
        \node[style={fill=black,circle}] (6) at (1.06,1.06){};
        \node[style={fill=black,circle}] (7) at (-1.061,-1.06){};
        \node[style={fill=black,circle}] (8) at (-1.06,1.06){};
        \node[style={fill=black,circle}] (9) at (1.06,-1.06){};
        \draw[black,very thick] (1)--(2) (1)--(3) (1)--(4) (1)--(5) (1)--(6) (1)--(7) (1)--(8) (1)--(9);
        
        \node[style={fill=black,circle}] (a) at (4,1.5){};
        \node[style={fill=black,circle}] (b) at (2.700961894,0.75){};
        \node[style={fill=black,circle}] (c) at (2.939339828,-1.06){};
        \node[style={fill=black,circle}] (d) at (5.06066,-1.06){};
        \node[style={fill=black,circle}] (e) at (5.299,0.75){};
        \draw[black,very thick] (a)--(b) (a)--(c) (a)--(d) (a)--(e) (b)--(c) (b)--(d) (b)--(e) (c)--(d) (c)--(e) (d)--(e);

        \node[style={fill=black,circle}] (1) at (7.5,0){};
        \node[style={fill=black,circle}] (2) at (7.5,1.5){};
        \node[style={fill=black,circle}] (3) at (8.66066,-1.06066){};
        \node[style={fill=black,circle}] (4) at (6.4393398,-1.06066){};
        \node[style={fill=black,circle}] (5) at (7.5,-2){};
        \draw[black,very thick] (1)--(2) (1)--(3) (1)--(4) (2)--(3) (2)--(4) (3)--(4) (3)--(5) (4)--(5);
        \draw (2) to [out=360,in=360,looseness=1.5] (5);

        \node[style={fill=black,circle}] (1) at (12,1){};
        \node[style={fill=black,circle}] (2) at (11.13397,0.5){};
        \node[style={fill=black,circle}] (3) at (11.29289,-0.7071){};
        \node[style={fill=black,circle}] (4) at (12.7071,-0.7071){};
        \node[style={fill=black,circle}] (5) at (12.866,0.5){}; 
        \node[style={fill=black,circle}] (6) at (12,2){};
        \node[style={fill=black,circle}] (7) at (13.732,1){};
        \node[style={fill=black,circle}] (8) at (13.4142,-1.414){};
        \node[style={fill=black,circle}] (9) at (10.26794919,1){};
        \node[style={fill=black,circle}] (10) at (10.585786,-1.414){};            
        \draw[black,very thick] (1)--(6) (1)--(4) (1)--(3) (2)--(9) (2)--(4) (2)--(5) (3)--(10) (4)--(8) (3)--(5) (5)--(7) (6)--(9) (6)--(7) (7)--(8) (8)--(10) (9)--(10);
        \end{tikzpicture} 
    \end{center}
        \caption{Example of average hereditary graphs}  
        \label{fig:exampl_ahg}     
\end{figure}
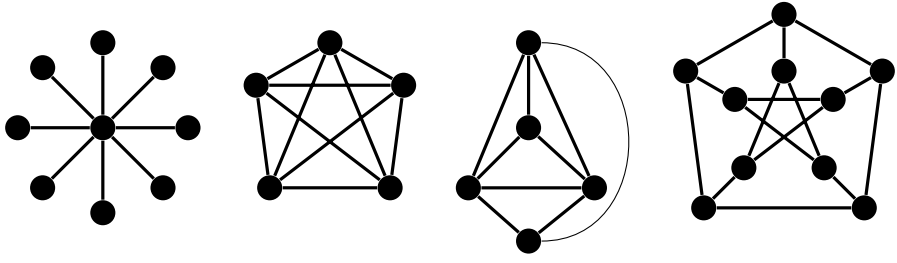

Some non-example of average hereditary graphs would be $K_7 \cup K_3$ or if we connect $P_{100}$ to a vertex of $K_5$. So these non-examples are graphs which has one region with a much higher edge density and one region with a much lower edge density. But these examples are much rarer, so the class of average hereditary graphs is much broader as the class is a superset of many popular classes of graphs. Figure~\ref{fig:example_non_ahg} shows a non-example of average hereditary graphs.
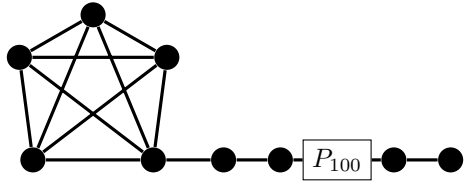
\begin{figure}[h]
    \vspace*{1cm}
    \begin{center}
    \begin{tikzpicture}[scale=.75]
    
        \node[style={fill=black,circle}] (1) at (4,1.5){};
        \node[style={fill=black,circle}] (2) at (2.700961894,0.75){};
        \node[style={fill=black,circle}] (3) at (2.939339828,-1.06){};
        \node[style={fill=black,circle}] (4) at (5.06066,-1.06){};
        \node[style={fill=black,circle}] (5) at (5.299,0.75){};
        \node[style={fill=black,circle}] (6) at (6.299,-1.06){};
        \node[style={fill=black,circle}] (7) at (7.299,-1.06){};
        \node[draw] (8) at (8.299,-1.06){$P_{100}$};
        \node[style={fill=black,circle}] (9) at (9.299,-1.06){};
        \node[style={fill=black,circle}] (10) at (10.299,-1.06){};
        \draw[black,very thick] (1)--(2) (1)--(3) (1)--(4) (1)--(5) (2)--(3) (2)--(4) (2)--(5) (3)--(4) (3)--(5) (4)--(5)
        (4)--(6) (6)--(7) (7)--(8) (8)--(9) (9)--(10);
        
    \end{tikzpicture} 
\end{center}        
\caption{Example of Non-average hereditary graphs}   
\label{fig:example_non_ahg}
\end{figure}

We can see that these average hereditary graphs are graphs which has more ``well'' distributed edge density. So we don't have that one region has a much higher edge density while another has a much lower edge density than that, as, if that were the case then removing the region with much lower edge density would give us a subgraph that has a higher average degree than the original graph.
Now we explore some interesting properties regarding average hereditary graphs.

\section{Computing Average Hereditary Property}\label{sec: computing}
At first glance, determining whether a graph is average hereditary appears to be a computationally challenging task. It is therefore of interest to find an equivalent condition that is efficiently computable.

Our next proposition gives that condition. Through which we will be able to decide the average hereditary property in polynomial time. 
The notion of maximum average degree has been an active area of research in graph theory. For a graph $G$, if $H$ is a subgraph of $G$ such that for all subgraphs $G'$ of $G$, the average degree of $G'$ is less than or equal to the average degree of $H$. For $G$ we denote $\MAD(G)$ as the average degree $H$. In other words for a graph $G$, $\MAD(G)$ is the average degree of the densest subgraph of $G$. The next proposition gives us a relation between $\MAD(G)$ and $d(G)$ for an average hereditary graph. 

\begin{proposition}\label{equivalent}
    A graph $G$ is average hereditary if and only if $d(G) = \MAD(G)$.
\end{proposition}
\begin{proof}
    Let $G = (V,E)$ be a graph, we know that $\MAD(G) \geq d(G)$. 
    Now suppose $G$ is average hereditary then from the definition $d(H) \leq d(G)$ for all subgraphs $H$ of $G$.
    Then if $\MAD(G) > d(G)$ we would have that there exists some subgraph $H$ of $G$ such that $d(H) > d(G)$ which is a contradiction, therefore $d(G) = \MAD(G)$.
    Conversely, suppose that $\MAD(G) = d(G)$ then from the definition of maximum average degree, $d(H) \leq \MAD(G) = d(G)$ for all subgraphs $H$ of $G$, therefore $G$ is average hereditary.
\end{proof}
In his 1984 paper Goldberg showed that $\MAD(G)$ can be computed in polynomial time \cite{12}.  
Goldberg reduced computing $\MAD(G)$ to a bunch of network flow problems giving an algorithm that runs in $O(M(|V|, 7 + |E|)\log |V|)$ time where $M(n, m)$ is the time required to find the minimum capacity cut in a network with $n$ vertices and $m$ edges \cite{12}.
The Max-flow min-cut theorem states that the minimum capacity of the cut in a network is equal to the maximum flow in a network \cite{1,2,8,14}. 
Now through this, we have several polynomial time algorithms to find the minimum capacity cut in a network such as Ford-Fulkerson algorithm \cite{8}, Edmonds-Karp algorithm \cite{13}, Dinitz's algorithm \cite{16}, Goldberg-Tarjan algorithm \cite{15} and many others. 
Furthermore, as the average degree of any graph can be computed in polynomial time we can use Goldberg's algorithm to compute if a given graph is average hereditary in polynomial time.

\section{Bound on Chromatic number}\label{sec: bound}
Graph coloring is famously $\NPC$.
Chromatic coloring is the optimization version of the Graph coloring problem. Chromatic coloring concerns with computing the chromatic number of a given graph. Unlike the usual graph coloring problem, which isn't specifically concerns with the chromatic number of the graph, the chromatic coloring problem isn't even known to be in \classNP. Chromatic coloring is itself \NPH. However many upper bounds exist for the chromatic number of a graph, aiming to reduce the complexity of the existing coloring algorithms. 
We introduce a new upper bound for the chromatic number of a graph in terms of the maximum average degree of the graph, that is $\chi(G) \leq \MAD(G)+1$. As a corollary of this, we obtain a relationship between the average degree and the chromatic number of an average hereditary graph. We compare our new upper bound on the chromatic number with previous bounds found in the literature and show that our bound is tighter than the previous bounds.
To prove this bound we use two Lemmas found in literature.

\begin{lemma}[\cite{2}]
    Every $k$-chromatic graph has a $k$-critical induced subgraph. 
\end{lemma}

\begin{lemma}[\cite{2}]
    For a critical graph $G$, $\chi(G) - 1 \leq \delta(G)$.  
\end{lemma}

\begin{theorem}\label{bound}
    For a graph $G$, $\chi(G) \leq \MAD(G)+1$.
\end{theorem}

\begin{proof}
    Let $G = (V,E)$ be a graph. 
    Let $H \subseteq G$ be a critical induced subgraph of $G$, then $\chi(H) = \chi(G)$. As $\chi(G) - 1 \leq \delta(H)$, then we have that $\chi(G) \leq \delta(H) + 1$. So we have that $\chi(G) \leq \delta(H) + 1 \leq d(H)+ 1 \leq \MAD(G)+ 1$. Therefore $\chi(G) \leq \MAD(G)+ 1$.
\end{proof}

\begin{corollary}\label{avgbound}
    If $G$ is an average hereditary graph then $\chi(G) \leq d(G)+1$.
\end{corollary}
\begin{proof}
    Let $G = (V,E)$ be an average hereditary graph. Then from Proposition~\ref*{equivalent} we have that $d(G) = \MAD(G)$, and so $\chi(G) \leq \MAD(G)+1$ implies $\chi(G) \leq d(G)+1$.
\end{proof}

It is to note that as $\chi(G) \in \mathbb{N}$, $\chi(G) \leq \MAD(G)+1$ implies $ \chi(G) \leq \lfloor \MAD(G)+1 \rfloor$ and likewise $\chi(G) \leq d(G)+1$ implies $ \chi(G) \leq \lfloor d(G)+1 \rfloor$.
The following corollary uses this to show a case where we know the exact chromatic number of an average hereditary graph.

\begin{corollary}\label{lowerequpper}
    If $G$ is an average hereditary graph and $\left\lceil\frac{|V|}{|V|-d(G)} \right\rceil = \lfloor d(G)+1 \rfloor$, then $\chi(G) = \omega(G) = \lfloor d(G)+1 \rfloor$.
\end{corollary}

\begin{proof}
    Let $G = (V,E)$ be an average hereditary graph. From \cite{4} we know that $\chi(G) \geq\omega(G)\geq\frac{|V|^2}{|V|^2-2|E|} = \frac{|V|}{|V| - d(G)}$.
    
    From corollary~\ref{avgbound} we have that $\chi(G) \leq d(G)+1$. As $\chi(G),\omega(G) \in \mathbb{N}$, then $\chi(G) \leq d(G)+1$ implies $\chi(G) \leq \lfloor d(G)+1 \rfloor$ and $\chi(G) \geq\omega(G) \geq \frac{|V|}{|V| - d(G)}$ implies $\chi(G) \geq \omega(G) \geq\left\lceil\frac{|V|}{|V|-d(G)} \right\rceil$. So 
    $$ \left\lceil\frac{|V|}{|V|-d(G)} \right\rceil \leq\omega(G) \leq \chi(G) \leq \lfloor d(G)+1 \rfloor$$ 
    Therefore if $\left\lceil\frac{|V|}{|V|-d(G)} \right\rceil = \lfloor d(G)+1 \rfloor$, then $\chi(G) = \omega(G) = \left\lceil\frac{|V|}{|V|-d(G)} \right\rceil = \lfloor d(G)+1 \rfloor$.
\end{proof}
We now show that our new bound of $\chi(G) \leq \MAD(G)+1$ is an improvement on previous bounds, that is to say, that in most cases our bound will be tighter than previous bounds found in the literature. We compare our bound with three previous bounds found in the literature showing that our bound is in improvement.

It is well known that in general $\chi(G) \leq \Delta(G) + 1$. We compare our bound with this general case bound and show that our bound is an improvement.

\begin{proposition}\label{improv1}
    $\lfloor \MAD(G)+1 \rfloor \leq \Delta(G) + 1$
\end{proposition}

\begin{proof}
    For a graph $G = (V,E)$ the maximum average degree $\MAD(G)$ is less than or equal to the maximum degree $\Delta(G)$. So
    $$\MAD(G) \leq \Delta(G) \implies \MAD(G)+1\leq \Delta(G) + 1 \implies \lfloor \MAD(G)+1 \rfloor \leq \Delta(G) + 1.$$
\end{proof}

Therefore this new bound is an improvement on the general case bound of $\chi(G) \leq \Delta(G) + 1$.

Now we compare it to the special case bound given by Brooks in 1941. Brooks showed that $\chi(G) \leq \Delta(G)$ if $G$ is neither complete nor an odd cycle \cite{7}. We show that our bound is an improvement on Brooks's bound when $G$ is not regular.

\begin{proposition}\label{improv2}
    $\lfloor \MAD(G)+1 \rfloor \leq \Delta(G)$ if not all connected component of $G$ containing a subgraph of maximum average degree are regular. 
\end{proposition}

\begin{proof}
Let $G = (V,E)$ be a graph such that there is some non-regular connected component of $G$ containing a subgraph of maximum average degree. 
Let $M$ be the subgraph of $G$ with the maximum average degree so $\MAD(G) = d(M)$ and Let $H$ be the connected component of $G$ containing $M$.
Without loss of generality suppose that there is a unique component $H$ of $G$ such that $H$ contains a graph with the maximum average degree, else if there are multiple such components of $G$ then we choose the component with the greatest maximum degree and obtain a graph $G'$ by deleting all other such components of $G$, this way $\Delta(G') = \Delta(G)$ and $\MAD(G') = \MAD(G)$.

We have that $\MAD(G) \leq \Delta(G)$, now suppose $\MAD(G) = \Delta(G)$. We know that $d(M) \leq \Delta(H) \leq \Delta(G)$.
So $\MAD(G) = \Delta(G)$ implies $\MAD(G) = \Delta(H)$. Now if  $\MAD(G) = \Delta(H)$, then as $\MAD(G) \leq \Delta(M)$, we have that $d(M) = \MAD(G) = \Delta(M) = \Delta(H)$, so all the vertices in $M$ have the same degree in $M$ therefore $M$ would be regular. Now suppose that $U = H - M \neq \emptyset$, then as $H$ is connected $[M,U]$ is nonempty. 
Then there is at least one vertex in $M$ that has one of its incident edges in $[M,U]$, but as we have that all vertices in $M$ have degree equal to $\Delta(H)$ in $M$ then $\MAD(G) < \Delta(H)$ which is a contradiction. So we have that $U = H - M = \emptyset$. Therefore $H$ is regular, which is a contradiction as the connected component of $G$ containing $M$ is not regular, therefore $\MAD(G) < \Delta(G)$.
As $\Delta(G) \in \mathbb{N}$, then $\lfloor \MAD(G) \rfloor < \Delta(G)$. As $\lfloor \MAD(G)\rfloor \in \mathbb{N}$, $\lfloor \MAD(G)\rfloor < \Delta(G) \implies \lfloor \MAD(G) \rfloor + 1 \leq \Delta(G)$.
\end{proof}

Therefore when not all connected components of $G$ containing a subgraph of maximum average degree are regular our bound is an improvement on the bound given by Brooks. 

Mar\'ia Soto, André Rossi and Marc Sevaux showed that $\chi(G) \leq \left\lfloor \frac{3+\sqrt{9+8(|E|-|V|)}}{2} \right\rfloor$ if $G$ is simple and connected \cite{3}. They compared their bounds with some previous bounds found in literature and showed that it was an improvement on those bounds \cite{3}.
We now compare our bound with this bound and show that our bound is an improvement on their bound as well. We first prove a lemma which we will use to prove that our bound is an improvement on the bound by Mar\'ia Soto, André Rossi and Marc Sevaux.

\begin{lemma}\label{lemImpov}
    For a simple and connected graph $G$, $\lfloor d(G)+1 \rfloor \leq \left\lfloor \frac{3+\sqrt{9+8(|E|-|V|)}}{2} \right\rfloor$.
\end{lemma}

\begin{proof}
    Let $G = (V,E)$ be a simple and connected graph, then $|V|-1 \leq |E| \leq \frac{|V|(|V|-1)}{2}$. If $|V| = 0$ or $|V|=1$ then $d(G) = 0$ so this is trivial case. Now we consider when $|V| > 1$.
    We first show that $d(G)+1 = \frac{2|E|}{|V|} + 1 \leq  \frac{3+\sqrt{9+8(|E|-|V|)}}{2}$ 

    We show that the inequality holds for $|V| \leq |E| \leq \frac{|V|(|V|-1)}{2}$. We will then show for $|E| = |V|-1$; $\lfloor d(G)+1 \rfloor = \left\lfloor \frac{3+\sqrt{9+8(|E|-|V|)}}{2} \right\rfloor =2$ where $|V| \geq 2$.
    $$\frac{2|E|}{|V|} + 1  \leq  \frac{3+\sqrt{9+8(|E|-|V|)}}{2}$$
    $$2|E|^2 - (|V|+|V|^2)|E| + (|V|^3-|V|^2) \leq 0$$
   
    $$\iff \frac{|V|+|V|^2-|V|\sqrt{9-6|V|+|V|^2}}{4} \leq |E| \leq \frac{|V|+|V|^2+|V|\sqrt{9-6|V|+|V|^2}}{4}$$
    We have that $|V| \leq |E| \leq \frac{|V|(|V|-1)}{2}$, so first we show that 
    $\frac{|V|(|V|-1)}{2} \leq \frac{|V|+|V|^2+|V|\sqrt{9-6|V|+|V|^2}}{4}$
    $$\frac{|V|(|V|-1)}{2} \leq \frac{|V|+|V|^2+|V|\sqrt{9-6|V|+|V|^2}}{4} \iff 9-6|V|+|V|^2 \leq 9-6|V|+|V|^2 \iff 0 \leq 0$$
   
    Now we show that $\frac{|V|+|V|^2-|V|\sqrt{9-6|V|+|V|^2}}{4} \leq |V|$
    $$\frac{|V|+|V|^2-|V|\sqrt{9-6|V|+|V|^2}}{4} \leq |V| \iff 9-6|V|+|V|^2 \leq 9-6|V|+|V|^2 \iff 0 \leq 0$$
   
    So we have that for $|V| \leq |E| \leq \frac{|V|(|V|-1)}{2}$, \;$d(G)+1 = \frac{2|E|}{|V|} + 1 \leq  \frac{3+\sqrt{9+8(|E|-|V|)}}{2}$.
    
    Now we will then show for $|E| = |V|-1$, $\lfloor d(G)+1 \rfloor = \left\lfloor \frac{3+\sqrt{9+8(|E|-|V|)}}{2} \right\rfloor = 2$.
    
    For \;$|E| = |V|-1$, $\frac{3+\sqrt{9+8(|E|-|V|)}}{2} = \frac{3+\sqrt{9+8(|V|-1-|V|)}}{2} = 2$ \;and 
    \;$|E| = |V|-1$, $\lfloor d(G)+1 \rfloor = \lfloor \frac{2|E|}{|V|} \rfloor + 1  = \lfloor \frac{2|V| - 2}{|V|} \rfloor + 1$.
    as $|V| \in \mathbb{Z}^+$ and $|V|$ is finite and $|V| \geq 2$ then $\lfloor \frac{2|V| - 2}{|V|} \rfloor = 1$
    so $\lfloor d(G)+1 \rfloor = 2$.
  
    Therefore for $|V|-1 \leq |E| \leq \frac{|V|(|V|-1)}{2}$,  $\lfloor d(G)+1 \rfloor \leq \left\lfloor \frac{3+\sqrt{9+8(|E|-|V|)}}{2} \right\rfloor$.
    
\end{proof}

\begin{proposition}\label{improv3}
    $\lfloor \MAD(G)+1 \rfloor \leq \left\lfloor \frac{3+\sqrt{9+8(|E|-|V|)}}{2} \right\rfloor$ if $G$ is simple and connected.
\end{proposition}

\begin{proof}
    Let $G = (V,E)$ be a simple and connected graph, and let $M$ be the subgraph of $G$ with the maximum average degree.
    We first show that following inequality holds:
    \begin{equation}\label{eqn:inequality_prop_floor}
        \left\lfloor \frac{3+\sqrt{9+8(|E(M)|-|V(M)|)}}{2} \right\rfloor \leq \left\lfloor \frac{3+\sqrt{9+8(|E|-|V|)}}{2} \right\rfloor.
    \end{equation}
    We first prove this by showing that this inequality holds even when we remove the floor functions. That is, 
    \begin{equation}\label{eqn:inequality_prop}
    \frac{3+\sqrt{9+8(|E(M)|-|V(M)|)}}{2} \leq \frac{3+\sqrt{9+8(|E|-|V|)}}{2}.
    \end{equation}
    
    From (\ref{eqn:inequality_prop}) we have that, $|E(M)|-|V(M)| \leq |E|-|V|$ that is $|V|-|V(M)| \leq |E|-|E(M)|$. 
    Let $H = G-M$ then $|V|-|V(M)| = |V(H)|$ and $|E|-|E(M)| = |E(H)| + |[H,M]|$. Now if $M = G$ then $|V|-|V(M)| = |E|-|E(M)| = 0$. Now if $M \neq G$, so $|V(H)| > 0$ and  $|E(H)| + |[H,M]| > 0$. Then (\ref{eqn:inequality_prop}) gives $|V(H)| \leq |E(H)| + |[H,M]|$. Let $S$ be the union of all single vertex components of $H$, and let $U = H - S$ so $|V(H)| = |V(U)| + |V(S)|$ and $|E(H)| + |[H,M]| = |E(U)| + |[U,M \cup S]| + |[S,M \cup U]|$. As $S$ is the union of single vertex components of $H$ and $G$ is connected we have that $|V(S)| = |[S,M \cup U]|$.
    So from (\ref{eqn:inequality_prop}) we get $|V(H)| \leq |E(H)| + |[H,M]| \iff |V(U)| \leq |E(U)| + |[U,M]|$.

    Now suppose that $|V(U)| > |E(U)| + |[U,M]|$. First, as $G$ is connected and $U$ contains no single vertex components then $E(U) \geq |V(U)| - 1$. Next as $G$ is connected $|[U,M]| \geq 1$, so $|V(U)| > |E(U)| + |[U,M]| \implies |V(U)| > |E(U)| + 1$ which is equivalent to $|V(U)| - 1 > |E(U)|$ which is a contradiction.  So we have that $|V(U)| \leq |E(U)| + |[U,M]|$ and therefore (\ref{eqn:inequality_prop}) holds which implies (\ref{eqn:inequality_prop_floor}).
    
    Now from Lemma~\ref*{lemImpov} we have that $\lfloor d(M)+1 \rfloor \leq \left\lfloor \frac{3+\sqrt{9+8(|E(M)|-|V(M)|)}}{2} \right\rfloor$, and as $d(M) = \MAD(G)$ we have that $\lfloor \MAD(G)+1 \rfloor \leq \left\lfloor \frac{3+\sqrt{9+8(|E(M)|-|V(M)|)}}{2} \right\rfloor \leq \left\lfloor \frac{3+\sqrt{9+8(|E|-|V|)}}{2} \right\rfloor$.
\end{proof}

Therefore our new bound of $\lfloor \MAD(G)+1 \rfloor$ is an improvement on the special case bound of $\chi(G) \leq \left\lfloor \frac{3+\sqrt{9+8(|E|-|V|)}}{2} \right\rfloor$ if $G$ is simple and connected. 

So we have a new bound on the chromatic number of a graph which we showed is an improvement on previous bounds. The bound is also computable in polynomial time as $\MAD(G)$ can be computed in polynomial time as discussed in section~\ref*{sec: computing}. The bound is also often very close to the actual value of the chromatic number.
Through this bound we found a relationship between the chromatic number and the average degree of an average hereditary graph, we also found a case for which we know the exact chromatic number of an average hereditary graph, this makes the class of average hereditary graphs interesting to us.

\section{NP-Hardness of Graph 3-Coloring in Average Hereditary Graphs}\label{sec: hardness}
In the previous section we explored some inequalities regarding the chromatic number of average hereditary graphs. We now explore the effects on the complexity of the graph \Prob{3-coloring} problem when we restrict the input to the class of average hereditary graphs. We show that graph \Prob{3-coloring} remains NP-Hard when the input is restricted to average hereditary graphs. This result is interesting as it gives us insight on the effect of input constraint on the graph \Prob{3-coloring} problem. This result can be used to study the complexity dichotomy of graph \Prob{3-coloring}, it can also be used to study the notion of boundary classes and limit class \cite{10,9}. We use the reduction given by Karp 1972 \cite{6} to prove our result, the proof then follows simply from the reduction.

\begin{lemma}[\cite{6}]
    \Prob{3-sat} $\leq_P$ \Prob{3-coloring}.
\end{lemma}

\begin{theorem}\label{nphardclass}
    Graph \Prob{3-coloring} is $\NPH$ for the class of average hereditary graphs.
\end{theorem}

\begin{proof}
    Let $\phi$ be a boolean expression in 3CNF with $L$ literals and $C$ clauses, where $L\geq 1$ and $C\geq 1$.  Let $G(\phi)$ denote the graph constructed from $\phi$ by the reduction given by Karp in \cite{6}. From the reduction we have that $G(\phi)$ has $6C+L+3$ vertices, and $12C+\frac{3}{2}L+3$ edges. We show that $G(\phi)$ is average hereditary.
    $$d(G(\phi)) = 3\left(\frac{8C+L+2}{6C+L+3}\right)$$
    Now we show that $d(G(\phi))< 4$. Suppose that $d(G(\phi))\geq 4$. So we have that $$\frac{8C+L+2}{6C+L+3} \geq \frac{4}{3} \iff 0 \geq \frac{1}{3}L + 2$$
    This is a contradiction as $C, L \in \mathbb{Z}^+$. Now for positive integers $x$, $y$, $a$ and $b$, we know that $\frac{x-a}{y-b} \leq \frac{x}{y} \iff \frac{a}{b} \geq \frac{x}{y}$. For $G(\phi)$, 
    $$d(G(\phi)) = \frac{\sum_{v\in V(G(\phi))}d_{G(\phi)}(v)}{|V(G(\phi))|}$$
    Now suppose $H$ is some induced subgraph of $G(\phi)$ such that $H = G - U$, for some $U \subseteq V(G(\phi))$, then for $H$,
    $$d(H) = \frac{\sum_{v\in V(G(\phi))}d_{G(\phi)}(v) - \sum_{u\in U}d_{G(\phi)}(u)-|[V(H),\overline{V(H)}]|}{|V(G(\phi))| - |U|} = \frac{\sum_{v\in V(G(\phi))}d_{G(\phi)}(v) - d}{|V| - |U|}$$
    As $G(\phi)$ is connected, and $\forall v \in V(G(\phi)),\; d_{G(\phi)}(v) \geq 2$, and no vertex with degree 2 are adjacent to each other, we have that $d  \geq 2(2|U|)$, which implies $\frac{d}{|U|} \geq 4$.
    So, $\frac{d}{|U|} \geq 4 \geq d(G(\phi)) \implies d(H) \leq d(G(\phi))$.
    Therefore $G(\phi)$ is average hereditary. So we have that for each boolean expression $\phi$ in 3CNF, the graph constructed $G(\phi)$, by the reduction given by Karp is average hereditary. So we have that graph \Prob{3-coloring} is $\NPH$ in the class of average hereditary graphs. 
\end{proof}

\section{Conclusion}
In this paper, we introduced this new class of graph. We initially introduced this class to obtain a tighter bound on the chromatic number of graphs in this class. Our initial aim was to find a case where we can bound the chromatic number of a graph in terms of its average degree. After obtaining our bound for average hereditary graphs we were able to generalize that bound and obtain a bound on the chromatic number of any given graph in terms of its maximum average degree.
We then explored the class of average hereditary graphs further. From that, we see that the class of average hereditary graphs itself is quite interesting. We looked at many interesting results regarding this class of graphs. Which includes an equivalent condition to compute if a graph belongs to this class and computational complexity of graph \Prob{3-coloring} when input is restricted to this class.

As the class of Average Hereditary graphs is new and just introduced in this paper there are a lot of opportunities for future work regarding the class of average hereditary graphs, such as exploring the class of graphs with other combinatorial optimization problems. This makes this class more interesting.

\section*{Acknowledgments}
The first author would like to thank Dr. \sloppy Hans Raj Tiwary (\url{https://kam.mff.cuni.cz/~hansraj/}) for providing with his invaluable expertise and guidance throughout the course of this research and for being a mentor in general. Many ideas for the direction to take this research also came from him. The first author would also like to thank Rameez Ragheb (\url{https://habib.edu.pk/SSE/rameez-ragheb/}) for being a mentor throughout the course of this research and prior. We would also like to thank him for verifying proofs of several results and helping in working out the algebra for Lemma~\ref{lemImpov}. Also would like to Acknowledge Meesum Ali Qazalbash (\url{https://www.linkedin.com/in/meesumaliqazalbash/}) for looking at the proof and algebra presented in this paper, verifying the results, and helping in the algebra for Lemma~\ref{lemImpov}. The authors would also like to thank Dr. Yair Caro for suggesting Proposition~\ref*{equivalent}, from this we were able to prove that the average hereditary property can be computed in polynomial time. This also gave us the idea to generalize our special case bound on the chromatic number to a general case bound on the chromatic number in terms of the maximum average degree.
\bibliographystyle{plain} 
\bibliography{graph_paper}

\end{document}